\documentclass{l4dc2025}

\title{Data-Driven and Stealthy Deactivation of Safety Filters}
\usepackage{times}
\usepackage{pgfplots}
\usepackage{tikz}
\usepackage{fontawesome5}
\usepackage{algorithm}
\usepackage[noend]{algpseudocode}
\usepackage{tcolorbox}

\algrenewcommand\textproc{\textsc}
\DeclareMathOperator*{\argmax}{arg\,max}
\DeclareMathOperator*{\argmin}{arg\,min}
\DeclareMathOperator*{\logdet}{log\,det}

\DeclareMathOperator{\inter}{int}

\author{%
 \Name{Daniel Arnström} \Email{daniel.arnstrom@it.uu.se}\\
 \addr Department of Information Technology, Uppsala University
 \AND
 \Name{André M.H. Teixeira} \Email{andre.teixeira@it.uu.se}\\
 \addr Department of Information Technology, Uppsala University
}

\begin{document}
\maketitle
\renewcommand{\baselinestretch}{1.0}

\definecolor{set19c1}{HTML}{E41A1C}
\definecolor{set19c2}{HTML}{377EB8}
\definecolor{set19c3}{HTML}{4DAF4A}
\definecolor{set19c4}{HTML}{984EA3}
\definecolor{set19c5}{HTML}{FF7F00}
\definecolor{set19c6}{HTML}{FFFF33}
\definecolor{set19c7}{HTML}{A65628}
\definecolor{set19c8}{HTML}{F781BF}
\definecolor{set19c9}{HTML}{999999}

\maketitle
\thispagestyle{empty}
\pagestyle{empty}
\newtheorem{assumption}{Assumption}
\newtheorem{problem}{Problem}

\SetKwBlock{Repeat}{repeat}{}
\newcommand{\icol}[1]{
    \left[\begin{smallmatrix}#1\end{smallmatrix}\right]%
}

\pgfplotstableread{result/iddata.dat}{\iddata}
\pgfplotstableread{result/latenttraj.dat}{\latenttraj} \pgfplotstableread{result/latensafeset.dat}{\latentsafeset}
\pgfplotstableread{result/attack.dat}{\attack}

\begin{abstract}
 Safety filters ensure that control actions that are executed are always safe, no matter the controller in question. Previous work has proposed a simple and stealthy false-data injection attack for deactivating such safety filters. This attack injects false sensor measurements to bias state estimates toward the interior of a safety region, making the safety filter accept unsafe control actions. The attack does, however, require the adversary to know the dynamics of the system, the safety region used in the safety filter, and the observer gain. In this work we relax these requirements and show how a similar data-injection attack can be performed when the adversary only observes the input and output of the observer that is used by the safety filter, without any a priori knowledge about the system dynamics, safety region, or observer gain. In particular, the adversary uses the observed data to identify a state-space model that describes the observer dynamics, and then approximates a safety region in the identified embedding. We exemplify the data-driven attack on an inverted pendulum, where we show how the attack can make the system leave a safe set, even when a safety filter is supposed to stop this from happening. 
\end{abstract}
\section{Introduction}

While learning-based controllers can improve performance over classical controllers \citep{coulson2019data,dorfler2023hottake}, they are seldom used in safety-critical applications due to their lack of \textit{safety guarantees}. So-called \textit{safety filters} \citep{wabersich2023data,tomlin2003computational,ames2017cbf,wabersich2018linear,hobbs2023runtime} can, in a modular fashion, augment any such unsafe learning-based controller with safety guarantees. 
A safety filter takes in a desired control action and the current state of the system, and outputs a filtered control action that guarantees a ``safe'' behaviour of the system. Since such filters separate safety from performance, any controller from the plethora of data-driven and learning-based controllers can be combined with a safety filter to give good performance together with safety guarantees.

\begin{figure}[H]
    \centering
    \begin{tikzpicture}[scale=0.8, transform shape]
    \usetikzlibrary{shapes,arrows}
    \usetikzlibrary{calc}
    \tikzstyle{block} = [
    rectangle, draw, fill=white!20, 
    text width=5.25em, text centered, 
    rounded corners, minimum height=3em]

    \node[block] (detect) at (-3,-1.05) {\faBell \:Detector};
    \node[block] (test1) at (0,1.05) {\faGamepad\:Primary \\Controller};
    \node[block] (test2) at (3.75,0) {\faShield* Safety \\ Filter};
    \node[block] (test3) at (7,0) {\faPlane\:Plant};
    \node[block] (observ) at (0,-1.05) {\faBinoculars\:Observer};
    \node (spy) at (4,-2.3) {\huge \faUserSecret};
    \node[red,right of=observ,xshift=1.5em,yshift=-1em] (adv-y) {$y^a$};
    \path [draw, -latex'] (test1) -- ++(2.4,0)  node[above]{$u_c$}-|(test2);
    \path [draw, -latex'] (test2) ++(1.6,0) |- ($(observ.east)+(0,+0.1)$);
    \path [draw, -latex'] (test2) --node[above]{$u$} (test3);
    \path [draw, latex'-] ($(observ.east)+(0,-0.1)$) -- ++(3,0) -| node[above right,yshift=0em]{$y$} (test3);
    \path [draw, -latex'] (observ) --node[left]{$\hat{x}$} (test1);
    \path [draw, -latex'] (observ) |- (test2);
    \path [draw, -latex'] (observ) --node[above]{$r$} (detect);
    \path[draw,-latex,red] (spy) -- node[right]{} ++(0,1);

\end{tikzpicture}
    \caption{\small Overview of the system architecture considered in \cite{ecc24} and in this paper. The safety filter produces a safe control action $u$ given a desired control $u_c$ and estimated state $\hat{x}$. An adversary tries to deactivate this filter through false-data injections on the communication channel between the sensors and the observer by replacing the true measurement $y$ with $y^a$.}
    \label{fig:overview}
\end{figure}
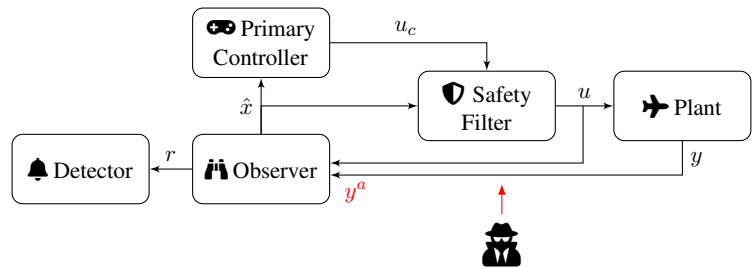

At the same time, cyber-physical systems (CPSs) are becoming more prevalent and enable advanced control systems that are efficient and resilient \cite{dibaji2019systems}. CPSs potential comes from their integration of communication, computation, and control technologies. The cyber component of CPSs do, however, open up for new vulnerabilities in the form of cyber attacks \citep{teixeira2015secure}\cite[\S 4.C]{roadmap}. 
Since safety filters are the last layer before a control command is applied, they are prime targets for cyber attacks; moreover, if such an attack successfully compromises a safety-critical system, the consequence can be severe \cite[\S 4.B]{roadmap}. 

In this paper we consider cyber attacks that target safety filters. In particular, we consider an attack that injects false data on the communication channel from sensor measurements to a state observer, as illustrated in Figure \ref{fig:overview}. The goal of the attack is to produce synthetic measurements $y^a$ that ``deactivate'' the safety filter, which in turn allows for dangerous control actions to be applied to the plant. 
The same attack scenario has been considered in \cite{ecc24}, but here we consider the far more realistic scenario when the adversary has less information. Specifically, in contrast to \cite{ecc24}, we assume that the adversary: ($i$) does not know the dynamics of the system; ($ii$) does not know the safety set that the safety filter uses; ($iii$) does not know the observer gain. Instead, the adversary learns all of these components by just observing data from the observer. 
In addition, the attack framework considered herein also applies to safety filters that are not based on control barrier functions (CBFs), while the work in \cite{ecc24} only considered CBF-based safety filters. 

\paragraph{Notation}
Given a set $\mathcal{S} \triangleq \{s : g(s) \geq 0\}$, defined as the super set of a function $g:\mathbb{R}^n \to \mathbb{R}^m$,  we let its boundary be defined as $\partial S \triangleq \{s : g(s) = 0\}$ and its interior be defined as $\inter S \triangleq \{s : g(s) > 0\}$. 
\section{Problem formulation}
In this paper we consider dynamical systems of the form 
\begin{equation}
    \label{eq:sys}
    x_{k+1} = f(x_k,u_k),
\end{equation}
with state $x \in \mathbb{R}^{n_x}$, control $u \in \mathcal{U}  \subseteq \mathbb{R}^{n_u}$, and  the dynamics $f : \mathbb{R}^{n_x}\times \mathbb{R}^{n_u} \to \mathbb{R}^{n_x}$.
Specifically, we are interested in the \textit{safety} of such systems. A system is safe if the states remain within a set $\mathcal{X}\subseteq \mathbb{R}^{n_x}$, where $\mathcal{X}$ encodes safety requirements. We formalize the notion of safety as follows
\begin{definition}[Safety]
    Given a set of admissible states $\mathcal{X}\subseteq \mathbb{R}^{n_x}$ and a set of admissible controls $\mathcal{U}\subseteq \mathbb{R}^{n_u}$, the system in \eqref{eq:sys} is said to be \textit{safe} if 
        $x_k \in \mathcal{X} \text{ and } u_k \in \mathcal{U} \text{ for all } k \geq 0.$
\end{definition}

Moreover, the system of interest generates measurements $y\in \mathbb{R}^{n_y}$ according to
\begin{equation}
 y_k = h(x_k) + e_k,
\end{equation}
with the measurement function $h: \mathbb{R}^{n_x} \to \mathbb{R}^{n_y}$ and measurement noise $e_k$; for our purpose, we make no particular assumptions on $e_k$. 

To control the system an estimate $\hat{x}$ of the state $x$ is needed. To this end, the controller uses an observer of the form 
\begin{equation}
    \label{eq:observer-dyn}
        \hat{x}_{k+1} = \hat{f}(\hat{x}_k,u_k,y_k),\qquad 
        \hat{y}_k = h(\hat{x}_k),
\end{equation}
with the goal of making the state estimate $\hat{x} \approx x$. 
A common form of the observer is 
\begin{equation}
    \label{eq:common-obs}
    \hat{f}(x,u,y) = f(x,u) + K(x) (y-h(x)),
\end{equation}
where $K: \mathbb{R}^{n_x} \to \mathbb{R}^{n_x \times n_y}$ is an observer gain. This form covers, for example, classical \citep{kalman1960new}, extended \citep{jazwinski2007stochastic}, and unscented \citep{julier2004unscented} Kalman filters. 

Furthermore, the controller uses an anomaly detector \citep{murguia2019model} that triggers an alarm if the measurement $y$ is too different from the predicted measurement $\hat{y}$. That is, we assume that the detector triggers an alarm if the residual $r \triangleq y-\hat{y} = y - h(\hat{x})$ is large, e.g., in terms of a cumulative sum \citep{kurt2018cusum} or a $\chi^2$-test \citep{brumback1987chi}. In particular, we consider anomaly detectors that trigger an alarm if  
\begin{equation}
    \label{eq:detector}
  \|r\| > \delta,
\end{equation}
where $\|\cdot\|$ is some norm and $\delta>0$ is a threshold that determines the sensitivity of the detector. 

Finally, the controller uses a \textit{safety filter} that ensures that states remain within a safe set $\mathcal{S} \subseteq \mathcal{X}$. This safe set is defined as the super-level set of a function $h_S$:
\begin{equation}
    \label{eq:safeset}
    \mathcal{S} \triangleq \{x : h_S(x) \geq 0\}.
\end{equation}
The set of control actions that retain the state in $S$ is defined as $\mathcal{U}_S(x) \triangleq \{u \in \mathcal{U} : f(x,u) \in \mathcal{S}\}$, and the corresponding safety filter is given by 
\begin{equation}
    \label{eq:sf}
    u = \argmin_{u\in \mathcal{U}_S(x)} \|u-u_c\|,
\end{equation}
where the desired control $u_c$ is projected onto the safe controls $\mathcal{U}_S$ to produce the actual control $u$ that is applied to the system.
Safety filters of the form \eqref{eq:sf} can be based on, for example, control barrier functions \citep{ames2017cbf}, predictive safety filters \citep{wabersich2021predictive}, or Hamilton-Jacobi reachability \citep{bansal2017hamilton}. Excellent surveys on safety filters are given in \cite{wabersich2023sf} and \cite{hobbs2023runtime}.

\subsection{Adversary model}
As mentioned in the introduction, and visualized in Figure~\ref{fig:overview}, we consider a scenario when an adversary performs a false-data injection attack on the measurement channel. Formally, we assume that the adversary has access to the following vulnerability: 
\begin{assumption}[Observe and modify measurements]
    \label{as:y-mod}
    The adversary can observe and modify $y$. 
\end{assumption}

The ultimate objective of the adversary is to make the system state leave the safe set, i.e., to make $x \notin \mathcal{S}$. If the observer in \eqref{eq:observer-dyn} perfectly estimates the state, that is, if $\hat{x} = x$, this would be impossible due to the safety filter. The adversary's aim is, hence, to bias the state estimate $\hat{x}$ to give a false sense of safety, which could allow for unsafe control actions to be let through the safety filter. This is formalized with the concept of \textit{deactivation} 
\begin{definition}[Deactivation]
    \label{def:deact}
    For a safety filter of the form \eqref{eq:sf}, the state estimate $\hat{x}$ is said to deactivate the safety filter at state $x$ if $\exists u \in \mathcal{U}_S(\hat{x})$ such that $u \notin \mathcal{U}_S(x)$.
\end{definition}

Briefly put, then, the adversary wants to modify the measurement $y$ such that the state estimate $\hat{x}$ deactivates the safety filter at the actual state.

If, in addition to Assumption \ref{as:y-mod}, the adversary knows the state estimate $\hat{x}$, the observer dynamics $\hat{f}$, and the measurement function $h$, they can, as is shown in \cite{ecc24}, perform a false-data injection attack by replacing the true measurement $y$ with
\begin{equation}
    \label{eq:ideal-attack}
        y^a = \argmax_{y^a}\: \nabla h_S(\hat{x})^T \left(\frac{\partial \hat{f}}{\partial y}\Bigr|_{x=\hat{x}} y^a\right) \text{subject to }\|y^a-h(\hat{x})\| \leq \delta.
\end{equation}
Intuitively, the objective of \eqref{eq:ideal-attack} is a local change in distance between $\hat{x}$ and the edge $\partial \mathcal{S} \triangleq \{x: h_S(x) = 0\}$ of the safe set $\mathcal{S}$; the constraint, in turn, ensures that the detector is not triggered, i.e., that the adversary remains stealthy. Thus, the attack will stealthily bias the state estimates towards the interior of the safe set, which make the safety filter overestimate the safety of the system, and might therefore let through unsafe control actions to the plant. For some norms, the optimization problem in \eqref{eq:ideal-attack} has a closed-form solutions (cf. Corollary 1 and 2 in \cite{ecc24}.)  

\begin{remark}[Independence of $\frac{\partial \hat{f}}{\partial y}$ with respect to $y$ and $u$]
    \label{rem:indep}
    Note that the term $\frac{\partial \hat{f}}{\partial y}\bigr|_{x=\hat{x}}$ in \eqref{eq:ideal-attack} implicitly assumes that the partial derivative of $\hat{f}$ with respect to $y$ is independent of $y$ and $u$. This is, for example, the case for observers of the common form \eqref{eq:common-obs}.  
\end{remark}

In contrast to \cite{ecc24}, this work considers the case when the observer dynamics in \eqref{eq:observer-dyn} and the safe set in \eqref{eq:safeset} are \textit{unknown} to the adversary, which make the \textit{ideal} attack in \eqref{eq:ideal-attack} impossible to implement. Instead, we assume that the adversary can observe the input ($y$, $u$) and output $\hat{y}$ of the observer, which motivates appending the following assumptions to Assumption~\ref{as:y-mod}.

\begin{assumption}[Observe control actions]
    \label{as:u}
    The adversary can observe $u$. 
\end{assumption}
\begin{assumption}[Observe predicted measurement]
    \label{as:yhat}
    The adversary can observe $\hat{y}$.
\end{assumption}
{\color{black}
Assumption \ref{as:yhat} can be satisfied in practice through several means. It is, for example, equivalent to observing the residual $r$ that is sent to the anomaly detector, since $\hat{y} = y- r$, with $y$ already being known from Assumption \ref{as:y-mod}.
Another, stronger, assumption that implies Assumption \ref{as:yhat} is that the adversary can directly observe the state estimate $\hat{x}$ (and knows the measurement function $h$.) Furthermore, if the adversary has some knowledge about the observer and controller structures they might perfectly estimate their internal states, which can give information about $r$, $\hat{y}$, and $\hat{x}$ \citep{umsonst2021control,umsonst2019observer}. 
}

In summary, the objective of the adversary, and the main problem formulation considered in this paper, is as follows: 
    \begin{tcolorbox}
        \textbf{Problem 1}: By observing the inputs $u$ and $y$ and the output $\hat{y}$ from  an observer of the form \eqref{eq:observer-dyn}, modify the measurement $y$ such that the state estimate $\hat{x}$ deactivates the safety filters of the form \eqref{eq:sf} at the true state $x$.
    \end{tcolorbox}

\section{Data-driven attack}
The main idea behind the adversary's attack is to replace the model-based components in \eqref{eq:ideal-attack} by surrogates that are learned using data from the system. For example, the adversary learns a mapping $\tilde{f}$ that predicts the observer output $\hat{y}$ given its input $u$ and $y$ to substitute the unknown observer dynamics $\hat{f}$. Furthermore, the collected data is used to identify a safe set similar to \eqref{eq:safeset}. Since the states used by the observer in \eqref{eq:observer-dyn} is assumed to be unknown, the learned state-space model will have states $z$ that do not necessarily correspond to $\hat{x}$. As a result, the safe set is identified in the \textit{latent} space where $z$ resides, and not the nominal state space where $\hat{x}$ resides.   

The learned model and safe set are then used to perform a false-data injection similar to \eqref{eq:ideal-attack}, where the unknown mappings $\hat{f}, h,$ and $h_S$ are replaced with learned mappings $\tilde{f}$, $\tilde{h}$, and $\tilde{h}_S$. To perform such an attack, the adversary also makes use of the learned dynamics and online data to update the latent state $z$ (instead of, as in the ideal attack in \eqref{eq:ideal-attack}, updating $\hat{x}$.)

The attack can, hence, be split into an \textit{online} phase and an \textit{offline} phase, which will now be described in detail in Section \ref{ssec:offline} and \ref{ssec:online}, respectively.

\subsection{Offline phase}
\label{ssec:offline}
In the initial \textit{offline} phase the adversary collects data by observing $u,y,$ and $\hat{y}$, which is used to:
    ($i$) identify a state-space model similar to \eqref{eq:observer-dyn};
    ($ii$) estimate a safe set similar to \eqref{eq:safeset} in latent space;
    ($iii$) lower-bound the detector threshold $\delta$;
Details on $(i)$, $(ii)$, and $(iii)$ are given in Section~\ref{sssec:sysid},\ref{sssec:safeset}, and \ref{sssec:thresh} respectively.

\begin{figure}[H] \centering
    \begin{tikzpicture}[scale=0.8, transform shape]
    \usetikzlibrary{shapes,arrows}
    \usetikzlibrary{calc}
    \usetikzlibrary{shapes.geometric}
    \tikzstyle{block} = [
    rectangle, draw, fill=white!20, 
    text width=15em, text centered, 
    rounded corners, minimum height=1em]


    \node[block,minimum height=6em] (data) at (-4,-1.65) {Collect data \\ $\mathcal{D} = \{(u_i,y_i,\hat{y}_i)\}_{i=1}^N$};
    \node[block] (sysid) at (3.5,-0.65) {Learn State-Space Model};
    \node[block] (sim) at (3.5,-1.75) {Estimate Safe Set in Latent Space};
    \node[block] (thresh) at (3.5,-2.5) {Lower-Bound Detector Threshold};

    \path [draw, -latex] (data.18) -- node[above,xshift=25pt]{} (sysid);
    \path [draw, -latex] (data.358) -- node[above]{} (sim);
    \path [draw, -latex] (data.344) -- node[below,xshift=25pt]{} (thresh);
    \path [draw, -latex] (sysid) -- node[right]{$\tilde{f}$} (sim);
    \path [draw, -latex] (sysid) -- node[above]{$\tilde{f}, \tilde{h}$} (8.5,-0.65);
    \path [draw, -latex] (sim) --node[above]{$\tilde{h}_S$} (8.5,-1.75);
    \path [draw, -latex] (thresh) --node[above]{$\tilde{\delta}$} (8.5,-2.5);

\end{tikzpicture}
    \caption{The pipeline for the offline phase of the data-driven attack.}
    \label{fig:pipeline}
\end{figure}
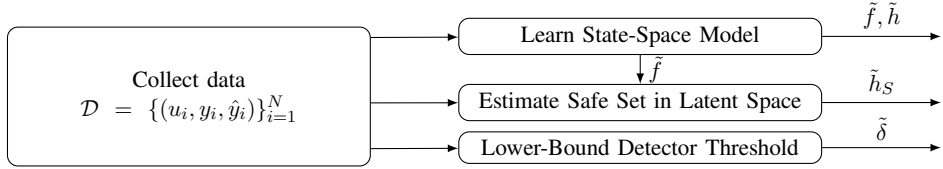

\subsubsection{System identification}
\label{sssec:sysid}
Assume that the adversary has observed the data sequence $\mathcal{D} = \{(u_i,y_i,\hat{y}_i)\}_i$ from \eqref{eq:observer-dyn}. 
To ensure that the detector is not to triggered when the measurement $y$ is modified, the adversary needs to have an idea where the controller's predicted measurement will be, i.e., it needs a predictive model that outputs $\tilde{y}$ such that $\tilde{y} \approx \hat{y}$. Given such a model, the adversary can stealthily modify the measurement $y$ by selecting it such that $\|y-\tilde{y}\|\leq\delta$ (if the threshold $\delta$ is known, see Section \ref{sssec:thresh}.)

The adversary is, hence, interested in learning a state-space model with a \textit{latent state} $z\in \mathbb{R}^{n_z}$ of the form
\begin{equation}
    \label{eq:latent-dyn}
    \begin{aligned}
        z_{k+1} = \tilde{f}(z_k,u_k,y_k), \qquad \tilde{y}_{k} = \tilde{h}(z_k).
    \end{aligned}
\end{equation}
such that $\tilde y_k \approx \hat{y}_k$ for all the data in $\mathcal{D}$. Note that $\tilde{y}\approx \hat{y}$ does not necessarily imply that the latent state $z$ is the same as the estimated state $\hat{x}$, which complicates the safe set estimation (see Section \ref{sssec:safeset}.)

There are several approaches to identifying models of the form \eqref{eq:latent-dyn}, including subspace methods \citep{VANOVERSCHEE199475}, Maximum Likelihood \citep{schon2011system}, deep learning \citep{gedon2021deep}, variational inference \citep{doerr2018probabilistic}, and direct nonlinear programming \citep{bemporad2024linear}, to name a few.
\begin{remark}[Motivation for state-space models]
    A predictive model that estimates $\hat{y}$ can be of other forms than the state-space model in \eqref{eq:latent-dyn} (For example, in the form of an ARX model, or from any other classical model classes in system identification \cite{ljung1999sysid}) . Having a state-space model will, however, be essential for the estimation of a safe set that is described in Section \ref{sssec:safeset}. 
\end{remark}

{\color{black}
Typically, we have that $\varepsilon_k \triangleq \hat{y}_k-\tilde{y}_k \neq 0$, since the $\tilde{f}$ never completely match $\hat{f}$ due to, for example, measurement and process noise. To ensure that $\tilde{y}$ remains close to $\hat{y}$ (i.e., to remain stealthy), one can amend the state-space model in \eqref{eq:latent-dyn} with an observer, resulting in the \textit{forward innovation model} \citep[Ch. 3]{van2012subspace} 
\begin{equation}
    \label{eq:latent-dyn-obs}
    \begin{aligned}
        z_{k+1} = \tilde{f}(z_k,u_k,y_k)+\tilde{K}(z_k) \varepsilon_k,\qquad \hat{y}_{k} = \tilde{h}(z_k) + \varepsilon_k,
    \end{aligned}
\end{equation}
where the gain $\tilde{K}(\cdot)$ also needs to be identified. This enables feedback from $\hat{y}$ during the online updates to ensure that $\hat{y} \approx \tilde{h}(z)$. For convenience, we introduce the compact notation $\tilde{f}(z,u,y,\hat{y}) \triangleq \tilde{f}(z,u,y)+\tilde{K}(z)(\tilde{h}(z)-\hat{y})$.
}

A well-studied (see, .e.g., \cite{van2012subspace}) special case of identifying state-space models of the form \eqref{eq:latent-dyn-obs} is when $\tilde{f}$ and $\tilde{h}$ are linear, which gives the linear state-space model
\begin{equation}
\label{eq:lin-ss}
  \begin{aligned}
      z_{k+1} = \tilde{A} z_k + \begin{bmatrix}
          \tilde{B}_u & \tilde{B}_y 
      \end{bmatrix} 
      \begin{bmatrix}
       u_k \\ y_k 
      \end{bmatrix} 
+  \tilde{K} \varepsilon_k, \qquad  
    \hat{y}_k = \tilde{C} z_k + \varepsilon_k.
  \end{aligned}
\end{equation}

Learning $\tilde{f}$ and $\tilde{h}$, hence, boils down to identifying the matrices $\tilde{A}, \tilde{B}_u,\tilde{B}_y, \tilde{K}$, and $\tilde{C}$. Note that we have that $\frac{\partial \tilde{f}}{\partial y} = \tilde{B}_y$, which is independent of both $u$ and $y$ in accordance with Remark~\ref{rem:indep}.

\begin{remark}[Inherent structure of $\hat{f}$]
    \label{rem:inh-struct}
Since we are are trying to identify the dynamics of an \textit{observer}, which is often based on simplified models that make the observer implementation manageable, the underlying dynamics is more structured than that of a ``natural'' system for which system identification is often applied.
\end{remark}

\subsubsection{Estimating safe set}
\label{sssec:safeset}
To deactivate the safe filter, the adversary wants to bias state estimates toward the interior of the safe set $\mathcal{S}$. Therefore, the adversary needs some information about the safe set $\mathcal{S}$. A complication is, however, that the adversary does not even know about the state space where $x$ resides, so ``estimating $\mathcal{S}$'' is not a well-defined problem given the available information. To address this complication, the adversary instead estimates a set in the space where the latent state $z$, from an identified model of the form \eqref{eq:latent-dyn}, resides. To connect these two state spaces, and the corresponding dynamics of $\hat{x}$ and $z$, we introduce the notion of topological equivalence, similar to what is introduced in \cite{jongeneel2021topological}. 

\begin{definition}[Topological equivalence]
    \label{def:top-equiv}
The dynamics $\hat{f}$ in \eqref{eq:observer-dyn} and $\tilde{f}$ \eqref{eq:latent-dyn} are \textit{topologically equivalent} if for all $u\in \mathbb{R}^{n_u}$ and $y \in \mathbb{R}^{n_x}$ there exists a homeomorphism $\mathcal{T}:\mathcal{X} \to \mathcal{Z}$ such that  $\tilde{f}(\cdot,u,y) \circ \mathcal{T} = \mathcal{T} \circ \hat{f}(\cdot,u,y)$, where $\circ$ denotes composition of functions.
\end{definition}

Basic topology gives the following results:
\begin{proposition}
    \label{prop:homeo}
    If $\tilde{f}$ and $\hat{f}$ are topologically equivalent through the homeomorphism $\mathcal{T}$, then \\
      $\partial \mathcal{T}(\mathcal{S}) = \mathcal{T}(\partial \mathcal{S})$ and 
      $\inter \mathcal{T}(\mathcal{S}) = \mathcal{T}(\inter \mathcal{S})$.
\end{proposition}

This shows that points in the interior of the safe set $\mathcal{S}$ directly map onto corresponding points in the interior of $\mathcal{T}(\mathcal{S})$. Consequently, if we bias the latent state $z$ into the interior of $\mathcal{T}(\mathcal{S})$, this will bias the state estimate $\hat{x}$ into the interior of $\mathcal{S}$. This connection is made more explicit in the following lemma.

\begin{lemma}[Safety of latent states]
    \label{lem:eqv-safe}
    Consider any sequence of inputs $\{u_i\}_{i=1}^N$ and outputs $\{y_i\}_{i=0}^N$ such that $u_i \in \mathcal{U}$ for all $i$. Assume that, subject to these sequences, the system \eqref{eq:observer-dyn} produces the trajectory $\{\hat{x}_i\}_{i=0}^N$ that is safe, i.e., $\hat{x}_i \in \mathcal{S}$ for all $i$. Similarly, assume and that system \eqref{eq:latent-dyn} produces the trajectory $\{z_i\}_{i=0}^N$ subject to the same input and output sequence, and with the initial state $z_0 = \mathcal{T}(\hat{x}_0)$. Then $z_i \in \mathcal{T}(\mathcal{S})$ if $f$ and $\tilde{f}$ are topologically equivalent through $\mathcal{T}$.
\end{lemma}
\begin{proof}
    From Definition \ref{def:top-equiv} we have that $\tilde{f}(\mathcal{T}(\hat{x}_k), u_k,y_k) = \mathcal{T}\left(\hat{f}(\hat{x}_k,u_k,y_k)\right) = \mathcal{T}(\hat{x}_{k+1})$. If it holds that $z_k = \mathcal{T}(\hat{x}_k)$ then we get that $z_{k+1} = \tilde{f}(z_{k},u_k,y_k) = \mathcal{T}(\hat{f}(\hat{x}_k,u_k,y_k))=\mathcal{T}(\hat{x}_{k+1})$. Since the base case $z_0 = \mathcal{T}(\hat{x}_0)$ holds, it follows by induction that $z_i = \mathcal{T}(\hat{x}_i),\quad \forall i =   0, \dots, N.$ 
    This together with the assumption that $\hat{x}_i \in \mathcal{S}$ gives the desired result that $z_i \in \mathcal{T}(\mathcal{S})$. 
\end{proof}

\paragraph{Linear special case}
When $\hat{f}$ is linear, the model class in \eqref{eq:lin-ss} contains several models that are topologically equivalent to $\hat{f}$ through a homeomorphism $\mathcal{T}$ that is linear ($\mathcal{T}$ is better known as a \textit{similarity transform} in this context.) Moreover, it is well known that subspace methods such as \texttt{N4SID} \citep{VANOVERSCHEE199475} identify one of these topologically equivalent models (cf., e.g., Chapter 2 in  \cite{van2012subspace}.)

\begin{remark}[Linear observers]
We want to emphasize that, related to Remark \ref{rem:inh-struct}, the special case of $\hat{f}$ being linear is quite common in practice since $\hat{f}$ is the observer dynamics that is designed by a user. For straightforward synthesis,  linearized models are commonly used for observer design. For example, this is the case when a classic Kalman filter is used.
\end{remark}

Lemma \ref{lem:eqv-safe} implies that latent trajectories give information about the underlying safe set $\mathcal{S}$ if the system is assumed to be safe (i.e., if the safety filter is deployed). Since, again, both $\mathcal{S}$ and $\mathcal{X}$ are unknown, we will leverage Lemma \ref{lem:eqv-safe} to approximate $\mathcal{T}(\mathcal{S})$ instead.  Let $\mathcal{Z}$ be a set that covers the latent trajectory $\{z_i\}_{i=1}^N$ generated by forward simulation with the data in $\mathcal{D}$. Moreover, let $\mathcal{Z}$ have an explicit representations as a super-level sets of the form $\mathcal{Z} = \{z : \tilde{h}_S(z) \geq 0 \}.$
While any cover of this form suffices for the attacker, two common covers are the convex hull of the trajectory, or a minimum covering ellipse. 

{\color{black}
\paragraph{Convex hull}
Given a trajectory $\{z_i\}_{i=1}^N$, a canonical way of getting a set that covers these points is as the convex hull $\mathcal{Z} = \text{Conv}{\{z_i\}_{i=1}}$, which can be put in the form of a super-level set $\{z : \tilde{h}(z) \geq 0\}$ by using its half-plane representation rather than its vertex representation.

\paragraph{Minimum covering ellipse} Another way of constructing $\mathcal{Z}$ is to find a matrix and vector $(Q,v)$ that define the ellipse $\{z :  \|Qz-v\|_2^2 \leq 1\}$ containing all of the data points in $\mathcal{D}$. That is, we want a matrix $Q \succ 0$ and translation $v\in \mathbb{R}^{n_z}$  such that  \begin{equation}
  \|Qz_i-v\|^2_2 \leq 1,\quad \forall z_i \in \mathcal{D},
\end{equation}
corresponding to $\tilde{h}_S(z) = 1-\|Qz-v\|^2_2$, and, in turn, $\nabla \tilde{h}_S(z) = 2(v-Qz)$.

The following semi-definite program (SDP) finds the smallest ellipsoid (in the $-\logdet$ sense) that covers all data points: 
\begin{equation}
    \label{eq:opt-ellipse}
        \min_{Q,v} -\logdet Q \quad \text{subject to } \|Q z_i -v \|^2_2 \leq 1,\:\: \forall z_i \in \mathcal{D}. 
\end{equation}
\begin{remark}[SOCP reformulation]
    Note that the problem in \eqref{eq:opt-ellipse} can be reformulated as a second-order cone program (SOCP), which improves the tractability of computing $\mathcal{Z}$. 
\end{remark}

So far we have not made any particular assumption on the shape of $\mathcal{S}$, other than that it can be described as a super-level set. In practice, however, the safety set $\mathcal{S}$ is often constructed as either a polyhedron or an ellipse.
In light of this, the following two lemmas relates $\mathcal{S}$ and $\mathcal{T}(\mathcal{S})$ for these cases when $\hat{f}$ is linear. The lemmas follows directly from the homeomorphism $\mathcal{T}$ being a similarity transform in the linear case.
\begin{lemma}
    Assume that $\hat{f}$ is linear and $\mathcal{S}$ is a polyhedron.  If $\tilde{f}$ is topologically equivalent to $\hat{f}$ through the homeomorphism $\mathcal{T}$, then $\mathcal{T}(S)$ is also a polyhedron.
\end{lemma}
\begin{lemma}
    Assume that $\hat{f}$ is linear and $\mathcal{S}$ is an ellipse.  If $\tilde{f}$ is topologically equivalent to $\hat{f}$ through the homeomorphism $\mathcal{T}$, then $\mathcal{T}(S)$ is also an ellipse.
\end{lemma}
}

This motivates using either a convex hull or a minimum covering ellipse for constructing $\mathcal{Z}$.

\subsubsection{Estimating detector threshold}
\label{sssec:thresh}
Even if the adversary can predict exactly where the observer's predicted measurement $\hat{y}$ is, it also needs to know the detector threshold $\delta>0$ to be able to know how much it can modify the measurement $y$. 
To ensure that modifications of $y$ are stealthy (i.e., that $\|y-\hat{y}\| < \delta$) the adversary, hence, needs to at least lower-bound the threshold $\delta$.

If we assume that the detector was not triggered while the data $\mathcal{D} = \{(u_i,y_i,\hat{y}_i)\}_{i}$ was collected,
\begin{equation}
    \label{eq:delta-bound}
    \|y-\hat{y}\| \leq \delta,\quad \forall (u,y,\hat{y}) \in \mathcal{D}.
\end{equation}
Hence, one way to find a lower-bound $\tilde{\delta}$ of $\delta$ is as
$\tilde{\delta} = \gamma \max_{(u,y,\hat{y})\in \mathcal{D}} \|y-\hat{y}\|,$
for some constant $\gamma \in (0,1]$; a smaller $\gamma$ reduce the risk of the adversary being detected, but also restricts how much the measurement $y$ can be modified. We formalize this in the following lemma the directly follows from \eqref{eq:delta-bound}.
\begin{lemma}
    Assume that the system uses a detector of the form \eqref{eq:detector} and that the detector was not triggered when collecting the data set $\mathcal{D}$. Then $\tilde{\delta}$ lower-bounds the detector threshold $\delta$(i.e., $\tilde{\delta} \leq \delta$.)
\end{lemma}

\subsection{Online phase}
\label{ssec:online}
In the online phase of the attack, the adversary performs a false-data attack by injecting false measurements $y^a$ according to the policy  

\begin{equation}
    \label{eq:id-attack}
    y^a(z_k) =  \argmax_{y^a}\: \nabla \tilde{h}_S(z_k)^T \left(\frac{\partial \tilde{f}}{\partial y}\Bigr|_{z=z_k} y^a\right) \text{subject to }           y^a-\tilde{h}(z_k)\| \leq \tilde{\delta},
\end{equation}
which is equivalent to the policy in \eqref{eq:ideal-attack}, but with an estimated system dynamics ($\tilde{f}$ and $\tilde{h}$), an estimated safe set ($\tilde{h}_S$), and a lower-bound $\tilde{\delta}$ of $\delta$. 
Moreover, the adversary updates the latent state $z$ by using the identified model of the form \eqref{eq:latent-dyn-obs} with $u$, $y^a$ and $\hat{y}$, resulting in the update
\begin{equation}
    z_{k+1} = \tilde{f}(z_k,u_k,y^a(z_k)) + \tilde{K}(z_k)(\tilde{h}(z_k) - \hat{y}_k) \triangleq \tilde{f}(z_k,u_k,y^a(z_k),\hat{y}_k).
\end{equation}
The attack is summarized in Algorithm~\ref{alg:attack}, and we will now make its relationship to the deactivation of safety filters more explicit.

{\color{black}
\subsubsection{Relationship between Algorithm 1 and deactivation of safety filters}
As mentioned before, the high-level idea behind the attack is to bias state estimates $\hat{x}$ toward the interior of $\mathcal{S}$. This is based on the set of safe control actions $\mathcal{U}_S(x)$ at $x$ often having the following property: $\mathcal{U}_S(x) \subset \mathcal{U}_S(\tilde{x})$ if $x$ is closer to $\partial \mathcal{S}$ than $\tilde{x}$, i.e., more control actions are safe further away from the border of the safe set. So if $x$ is closer to $\partial S$ than $\hat{x}$, we typically have that $\mathcal{U}_S(x) \subset \mathcal{U}_S(\hat{x})$, which, from Definition \ref{def:deact}, would deactivate the safety filter.

To relate this to the latent state $z$, Proposition \ref{prop:homeo} implies that biasing $\hat{x}$ into the interior of $\mathcal{S}$ is equivalent to biasing the latent state $z$ into the interior of $\mathcal{T}(\mathcal{S})$ if $\tilde{f}$ is topologically equivalent to $\hat{f}$ through the homeomorphism $\mathcal{T}$. Since the attacker does not have access to $\mathcal{T}(S)$ directly, they approximate it with $\mathcal{Z} = \{z : \tilde{h}_S(z) \geq 0\}$, which is a cover of points in $\mathcal{T}(S)$ (from Lemma \ref{lem:eqv-safe}).

To complete the argument, the measurement $y^a$ should bias the latent state $z$ towards the interior of $\mathcal{Z}$. Biasing $z$ towards the interior of $\mathcal{Z}$ means increasing the corresponding safety margin $\tilde{h}_S$. The measurement $y$ indirectly affects $\tilde{h}_S$ through the time update  $z_{k+1} = \tilde{f}(z_k,u_k,y)$. How a change in $y$ locally affects $\tilde{h}_S$ can be seen through the gradient $\nabla_y \left(h \circ \tilde{f}(z_k,u_k,y)\right)$. The chain rule then gives  

\begin{equation}
\nabla_y\left( h \circ \tilde{f}(z_k,u_k,y)\right) = \nabla \tilde{h}_S(z)^T \frac{\partial \tilde{f}}{\partial y}\Bigr|_{z=z_k},
\end{equation}
which is exactly the linear term in the objective of \eqref{eq:id-attack}. In other words, the objective of \eqref{eq:id-attack} aligns the false measurement $y^a$ in the direction that locally increases $\tilde{h}_S$ the most.
}

\begin{figure*}
{\centering
\begin{minipage}{.475\linewidth}
\begin{algorithm}[H]
    \caption{Data-driven false-data injection attack to deactivate safety filters.}
    \label{alg:attack}
    \textbf{Input:} $\tilde{f}, \tilde{h}, \tilde{h}_S,\tilde{\delta}$ (from offline phase) and $z_0$ \\
    $k \leftarrow 0$; \: $z \leftarrow z_0$\\
    \Repeat{
        observe $u_k$, $y_k$ and $\hat{y}_k$\\
            replace $y_k$ with $y^a$ by solving \eqref{eq:id-attack} for $z_k$ \\
        $z_{k+1} \leftarrow \tilde{f}(z_k,u_k,y_k,\hat{y}_k)$;\quad 
        $k \leftarrow k +1$
    }
\end{algorithm}
\end{minipage}
\hspace{1cm}
\begin{minipage}{.4\linewidth}
    \begin{figure}[H]
      \centering
      \begin{tikzpicture}[scale=0.8, transform shape]
          \usetikzlibrary{shapes,arrows,calc}
          \usetikzlibrary{calc}
          \usetikzlibrary{angles, quotes}
          \coordinate (hitch) at (0.81,1);
          \coordinate (anglehelp) at (0.81,3.3);
          \coordinate (pendulumtip) at (-0.5,3);
          \coordinate (uhelp) at (0.3,0.7);
          \draw[line width=2pt] (hitch) -- (pendulumtip);
          \draw[dashed] (hitch) -- (anglehelp);
          \draw[fill=gray] (pendulumtip) circle (3.5pt);
          \draw[fill=gray] (hitch) -- ++(-0.5,-0.5) -- ++(1,0) -- cycle;
          \draw[fill=black] (0.81,1) circle (1pt);
          \pic [red,draw,-latex, angle radius=11mm, angle eccentricity=1.2,thick] {angle = anglehelp--hitch--pendulumtip};
          \pic [blue,draw,-latex, angle radius=4mm, angle eccentricity=1,thick] {angle = pendulumtip--hitch--uhelp};

          \node[red] (test) at (0.425,2.3) {$\theta$};
          \node (mass) at (-0.1,3) {$m$};
          \node (length) at (-0.1,2) {$l$};
          \node[blue] (u) at (0.275,1.2) {$u$};
      \end{tikzpicture}
      \caption{\small Inverted pendulum with mass $m$ and length $l$. The angle $\theta$ is controlled by applying the torque $u$.}
      \label{fig:invpend}
    \end{figure}
\end{minipage}
}
\end{figure*}


\section{Numerical Experiments}
We exemplify the proposed attack on an inverted pendulum system that was also considered in \cite{wabersich2023sf} and \cite{alan2023control}, visualized in Figure \ref{fig:invpend}.
The dynamics of the system is 
    $\frac{\mathrm{d}}{\mathrm{d}t}
    \icol{
        \theta \\
        \dot{\theta}
    }
    = \icol{ 
        \dot{\theta} \\ 
        \frac{g}{l} \sin{\theta}
    }
    +\icol{
     0 \\
     \frac{1}{ml^2}
 }u,$ with parameter values $m=2$ kg, $l=1$ m , and $g=10$ m/s$^2$ (the same as in \citet{wabersich2023sf,alan2023control}). Moreover, 
the applied torque $u$ is limited between $\pm 3$ Nm, and the angle $\theta$ is constrained to be between $\pm 0.25$ radians out of safety reasons.
To ensure that these constraints are satisfied, a safety filter that retains the state in a safe set $\mathcal{S} = \{(\theta,\dot{\theta}): h_S(\theta,\dot{\theta}) \geq 0\}$ is used, where $h_S(\theta,\dot{\theta})$ is a control barrier function \citep{ames2017cbf}. Specifically, we use the control barrier function 
$h_S(\theta,\dot{\theta}) = 1 - 16 \theta^2 - 8\theta \dot\theta - 4 \dot{\theta}^2$
that is used in \cite{alan2023control}. The corresponding optimization problems in \eqref{eq:sf} are quadratic programs, which are solved using the quadratic programming solver \texttt{DAQP} \citep{arnstrom2022dual}. Code for all experiments is available at \url{https://github.com/darnstrom/ddsd-sf}.

Only the angle $\theta$ is measured, with measurement noise that is normally distributed with zero-mean and standard deviation $10^{-3}$ radians. The controller estimates $\dot{\theta}$ with an observer of the form \eqref{eq:observer-dyn}, where the gain $K$ is the static Kalman gain for a linearized model around the origin. 

For the attack, the adversary first collects data from the system by observing $(u,y,\hat{y})$ for 100 seconds. The collected data is shown in Figure~\ref{fig:iddata}.
{\centering
    \begin{figure*}
        \begin{minipage}{.5\textwidth}
            \vspace{0pt}
            \begin{figure}[H]
                \centering
                \begin{tikzpicture}[scale=1]
    \begin{axis}[
        xmin=0,xmax=100,
        ytick={-0.1,0.1},
        xticklabel=\empty,
        legend style={at ={(0.5,2.2)},anchor=north}, ymajorgrids,yminorgrids,xmajorgrids,
        x post scale=0.9,
        y post scale=0.1,
        legend cell align={left},legend columns=3,
        ]
        \addplot [set19c1,thick] table [x={t}, y={y}] {\iddata}; 
        \addlegendimage{set19c3, very thick};
        \addlegendimage{set19c2, very thick};
        \legend{$y$,$\hat{y}$, $u$};
    \end{axis}
\end{tikzpicture}
\begin{tikzpicture}[scale=1]
    \begin{axis}[
        xmin=0,xmax=100,
        ytick={-0.1,0.1},
        xticklabel=\empty,
        legend style={at ={(0.5,1.3)},anchor=north}, ymajorgrids,yminorgrids,xmajorgrids,
        x post scale=0.9,
        y post scale=0.1,
        legend cell align={left},legend columns=2,
        ]
        \addplot [set19c3,thick] table [x={t}, y={yh}] {\iddata}; 
    \end{axis}
\end{tikzpicture}
\begin{tikzpicture}[scale=1]
    \begin{axis}[
        xmin=0,xmax=100,
        ytick={-3,3},
        xlabel={Time [s]},
        legend style={at ={(0.5,1.2)},anchor=north}, ymajorgrids,yminorgrids,xmajorgrids,
        x post scale=0.9,
        y post scale=0.1,
        legend cell align={left},legend columns=2,
        ]
        \addplot [set19c2,thick] table [x={t}, y={u}] {\iddata}; 
    \end{axis}
    \path (-1.1575,0) -- (0,0);
\end{tikzpicture}
                \vspace{-20pt}
                \caption{Collected data $\mathcal{D}$ from the system, used for learning $\tilde{f}$, $\tilde{h}$, $\tilde{h}_S$ and $\tilde{\delta}$. 
                }
                \label{fig:iddata}
            \end{figure}
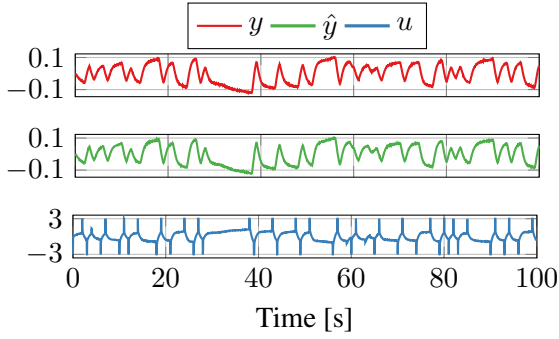
        \end{minipage}
        \quad
        \begin{minipage}{.5\textwidth}
            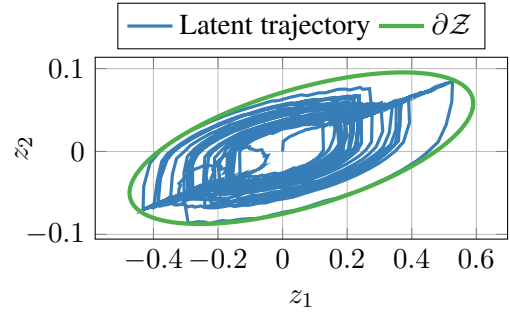
\begin{figure}[H]
                \centering
                \begin{tikzpicture}[scale=1]
    \pgfplotsset{compat = 1.3}
    \begin{axis}[
        ytick={-0.1,0,0.1},
        xlabel={$z_1$},
        ylabel={$z_2$},
        ylabel shift = -10 pt,
        legend style={at ={(0.5,1.3)},anchor=north}, ymajorgrids,yminorgrids,xmajorgrids,
        x post scale=0.8,
        y post scale=0.425,
        legend cell align={left},legend columns=2,
        ]
        \addplot [set19c2,very thick] table [x={x}, y={y}] {\latenttraj}; 
        \addplot [set19c3,ultra thick] table [x={x}, y={y}] {\latentsafeset}; 
        \legend{Latent trajectory, $\partial \mathcal{Z}$}
    \end{axis}
\end{tikzpicture}
                \caption{Latent trajectory from simulating \eqref{eq:latent-dyn} with the collected data $\mathcal{D}$.
                }
                \label{fig:latent}
            \end{figure}
        \end{minipage}
    \end{figure*}
}
Based on this data, the adversary identifies a linear state-space model of the form \eqref{eq:lin-ss}, with dimension $n_z=2$, using \texttt{N4SID} \citep{VANOVERSCHEE199475}. 

Using the identified state-space model, the adversary reuses the observed data $(u,y,\hat{y})$ and simulate the system, producing the latent trajectory $\{z_i\}$ shown in Figure~\ref{fig:latent}. An approximate safe set $\mathcal{Z} = \{z : \tilde{h}_S(z) \geq 0\}$ in latent space is then determined by finding a minimally covering ellipse for the entire latent trajectory, also shown in Figure~\ref{fig:latent}.

Next, the attacker performs the online phase of the attack, which uses the identified $\tilde{f}, \tilde{h}$ and $\tilde{h}_S$ to attack the system with Algorithm~\ref{alg:attack}. During the attack, the controller that produces desired control actions is an exact-linearized LQR controller given by 
\begin{equation}
    u_{c}(\theta,\hat{\theta}) = -m l g \sin({\theta}) - \begin{bmatrix}
    1.5 & 1.5 
\end{bmatrix} \begin{bmatrix}
    {\theta} \\ \dot{{\theta}} 
\end{bmatrix}.
\end{equation}
The attack scenario starts with $\dot{\theta} = 0.5$ and $\theta = -0.2$ at $t=0$.
To remain stealthy, the adversary does not initiate the attack directly at $t=0$; this is due to the latent state at $t=0$ being unknown, which, hence, requires some time for $z$ to converge such that $\tilde{h}(z) \approx \hat{y}$. The adversary initiate the attack at $t = 1$ second, which, as in shown in Figure \ref{subfig:trajectory}, drives the state $x$ outside the safe set $\mathcal{S}$ while the state estimates $\hat{x}$ remain inside $\mathcal{S}$, creating a false belief of safety.
Moreover, the corresponding residuals used by the anomaly detector are shown Figure \ref{subfig:resids}, where one can see that their magnitudes do not exceed the magnitude of the measurement noise. Note, however, that there is a clear bias in the residual after the attack. This can be used to create an anomaly detector that specifically detects attacks that try to deactivate safety filters, complementing classical detectors such as \eqref{eq:detector}. For details of how to construct such a detector, see \cite{ecc24}. 
\begin{figure}
  \centering
  \subfigure[Actual and percieved state trajectory \label{subfig:trajectory}]{%
      \pgfplotsset{compat=1.5.1}
\begin{tikzpicture}[scale=1]
    \begin{axis}[
        width=0.45\textwidth,
        height=0.275\textwidth,
        xmin=-0.6,xmax=0.6,
        ymin=-0.3,ymax=0.3,
        ylabel shift = -5 pt,
        xlabel={$\dot{\theta}$},
        ylabel={${\theta}$},
        legend style={at ={(0.5,1.25)},anchor=north}, ymajorgrids,yminorgrids,xmajorgrids,
        legend cell align={left},legend columns=3,
        axis equal,
        ]
        \filldraw[set19c3,fill opacity=0.2] 
        (axis cs:0,0) ellipse [y radius=0.5988,x 
        radius=0.2410,rotate=-106.845];
        \addplot [set19c2, very thick] table [x={x2}, y={x1}] {\attack};
        \addplot [set19c1, very thick, dashed] table [x={xh2}, y={xh1}] {\attack};
        \node[black] (atk-from) at (axis cs:-0.1138,0.18) {\scriptsize Attack starts};

        \draw[-latex] (axis cs:-0.1138,0.15) -- (axis cs:-0.1138,0.042); 

        \addlegendimage{area legend,fill=set19c3,opacity=0.2};
        \legend{$x$, $\hat{x}$,$\mathcal{S}$}
    \end{axis}
\end{tikzpicture}
  }
  \subfigure[Residual $r = y-\hat{y}$ \label{subfig:resids}]{%
      \begin{tikzpicture}[scale=1]
    \begin{axis}[
        xmin=0,xmax=2.15,
        ymin=-0.005,ymax=0.005,
        ylabel near ticks,
        xlabel={Time [s]},
        ylabel={$r$},
        legend style={at ={(0.5,1.25)},anchor=north}, ymajorgrids,yminorgrids,xmajorgrids,
        x post scale=0.9,
        y post scale=0.475,
        legend cell align={left},legend columns=3,
        ]
        \draw[red,thick,dashed] (axis cs:1,0.005)--(axis cs:1,-0.005);
        \addplot [set19c2,very thick] table [x={t}, y={r}] {\attack}; 
        \legend{Residual};
        \node[red] (atk) at (axis cs:1.225,-0.003) {\scriptsize Attack starts};
    \end{axis}
\end{tikzpicture}
}
\caption{The \textcolor{set19c2}{actual} and \textcolor{set19c1}{perceived} state trajectories, and the corresponding residual $r$, when the false-data attack in Algorithm \ref{alg:attack} is initiated at $t=1$. Consequently, the state leaves the \textcolor{set19c3}{safe set}.}
\end{figure}
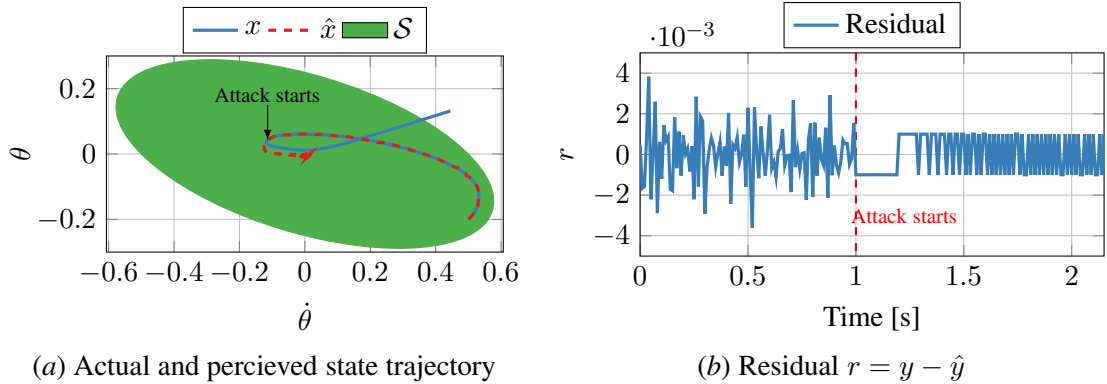

\section{Conclusion}
We have shown how a simple, data-driven, and stealthy false-data injection attack can deactivate safety filters. In contrast to the work in \cite{ecc24}, the adversary does not use any prior knowledge about the dynamics, safe set, or observer gain. The attack injects false sensor measurements to bias latent states to the interior of an estimated safety region in a latent embedding, which, in turn, biases state estimates to the interior of the actual safety region. This can make the safety filter let through unsafe control actions. 
The attack was shown to be able to successfully make a inverted pendulum system leave its safe set that a safety filter is supposed to keep invariant.


\acks{This work is supported by the Swedish Foundation for Strategic Research.}

\bibliography{lib.bib}

\end{document}